\documentclass{amsart}
\usepackage{amsmath}
\usepackage{amsthm}
\usepackage{listings}
\usepackage{xcolor}
\usepackage{mathtools}
\usepackage{paralist}
\usepackage[noadjust]{cite}
\usepackage{thmtools,thm-restate}
\usepackage{latexsym}
\usepackage{textcomp}
\usepackage{amssymb}
\usepackage{enumitem}  
\usepackage{hyperref}
\usepackage[mathscr]{euscript}
\usepackage{lineno}

\newtheorem{theorem}{Theorem}[section]
\newtheorem{proposition}[theorem]{Proposition}
\newtheorem{lemma}[theorem]{Lemma}
\newtheorem{corollary}[theorem]{Corollary}
\newtheorem{alphatheorem}{Theorem}

\theoremstyle{definition}
\newtheorem{definition}[theorem]{Definition}

\newtheorem{remark}[theorem]{Remark}

\newtheorem{question}[theorem]{Question}

\setenumerate{label={\normalfont(\textit{\roman*})}, ref={(\textrm{\roman*})},wide}

\newcommand{\ip}[1]{\left\lfloor #1 \right\rfloor }

\newcommand{\nint}[1]{\left\lfloor #1 \right\rceil}

\newcommand{\fpa}[1]{\left\lVert #1 \right\rVert_{\RR/\ZZ}}

\newcommand{\floor}[1]{\left\lfloor #1 \right\rfloor}

\newcommand{\bra}[1]{\left( #1 \right)}

\newcommand{\abs}[1]{\left|#1\right|}
\newcommand{\set}[2]{\left\{ #1 \ \middle| \ #2 \right\} }

\newcommand{\ceil}[1]{\left\lceil #1 \right\rceil}
 
\newcommand{\e}{\varepsilon}
\renewcommand{\a}{\alpha}

\newcommand{\NN}{\mathbb{N}}
\newcommand{\QQ}{\mathbb{Q}}
\newcommand{\ZZ}{\mathbb{Z}}
\newcommand{\RR}{\mathbb{R}}
 
\DeclareMathAlphabet{\mathpzc}{OT1}{pzc}{m}{it}

\newcommand{\mynote}[2]{\noindent{\bfseries\sffamily\scriptsize#1}{\small$\blacktriangleright$\textsf{\textsl{#2}}$\blacktriangleleft$}}
\newcommand\JK[1]{\mynote{JK}{{\color[rgb]{0.5,0.2,0.0} #1}}}
\newcommand\HB[1]{\mynote{HB}{{\color[rgb]{0.0,0.5,0.2} #1}}}

\definecolor{fresh}{HTML}{000000}\definecolor{checked}{HTML}{000000}
\definecolor{external}{HTML}{000000}\definecolor{later}{HTML}{000000}
\definecolor{minor-rev}{HTML}{000000}\definecolor{major-rev}{HTML}{000000}\definecolor{skip}{HTML}{000000}\definecolor{normal}{HTML}{000000}

\renewcommand{\subset}{\subseteq}

\newcommand{\DeltaS}{\operatorname{\mathrlap{\ \circ}\operatorname{\Delta}}}
\newcommand{\DeltaT}{\operatorname{\mathrlap{\ \phantom{\circ}}\operatorname{\Delta}}}
 
\newcommand{\theory}[1]{\mathrm{Th}(#1)}
\newcommand{\inv}[1]{#1^{-1}}
\newcommand{\defeq}{\overset{\mathrm{def}}{=}}

\newcommand*\patchAmsMathEnvironmentForLineno[1]{\expandafter\let\csname old#1\expandafter\endcsname\csname #1\endcsname
  \expandafter\let\csname oldend#1\expandafter\endcsname\csname end#1\endcsname
  \renewenvironment{#1}{\linenomath\csname old#1\endcsname}{\csname oldend#1\endcsname\endlinenomath}}\newcommand*\patchBothAmsMathEnvironmentsForLineno[1]{\patchAmsMathEnvironmentForLineno{#1}\patchAmsMathEnvironmentForLineno{#1*}}\AtBeginDocument{\patchBothAmsMathEnvironmentsForLineno{equation}\patchBothAmsMathEnvironmentsForLineno{align}\patchBothAmsMathEnvironmentsForLineno{flalign}\patchBothAmsMathEnvironmentsForLineno{alignat}\patchBothAmsMathEnvironmentsForLineno{gather}\patchBothAmsMathEnvironmentsForLineno{multline}}

\begin{document}
 
\author[H.\ Brown]{Hera Brown}
\address{Department of Computer Science, University of Oxford,
Wolfson Building, Parks Road, Oxford OX1 3QD, UK}
\email{hera.brown0707@gmail.com}

\author[J.\ Konieczny]{Jakub Konieczny}
\address{Department of Computer Science, University of Oxford,
Wolfson Building, Parks Road, Oxford OX1 3QD, UK}
\email{jakub.konieczny@gmail.com}

\renewcommand{\JK}[1]{}
\renewcommand{\HB}[1]{}
\renewcommand{\checkmark}{}
 
\title[Hardy field extensions of Presburger arithmetic]{Decidability of extensions of Presburger arithmetic by Hardy field functions}
\date{\today}
 
\begin{abstract}
We study the extension of Presburger arithmetic by the class of sub-polynomial Hardy field functions, and show the majority of these extensions to be undecidable. More precisely, we show that the theory $\theory{\ZZ; <, +, \nint{f}}$, where $f$ is a Hardy field function and $\nint{\cdot}$ the nearest integer operator, is undecidable when $f$ grows polynomially faster than $x$. Further, we show that when $f$ grows sub-linearly quickly, but still as fast as some polynomial, the theory $\theory{\ZZ; <, +, \nint{f}}$ is undecidable.
\end{abstract}
 
 
\maketitle 

\section{$\checkmark$Introduction}
\label{sec:intro}

\newcommand{\Th}{\mathrm{Th}}

\subsection{$\checkmark$Background}
Presburger arithmetic, the first-order theory $\theory{\NN; +}$ of the natural numbers with addition, is known to be decidable \cite{Presburger-1929}, whereas Peano arithmetic, the extension of Presburger arithmetic by multiplication $\theory{\NN; +, \times}$, is known to be undecidable \cite{Church-1936}. The undecidability of Peano arithmetic provides a method that can be used to show that various other extensions of Presburger arithmetic are undecidable; if multiplication can be defined in such an extension, then that extension is undecidable. For instance, the extension of Presburger arithmetic with the squaring function $\theory{\NN; +, x \mapsto x^2}$ is undecidable, since $a = bc$ holds \emph{iff} $a = (b+c)^2 - b^2 - c^2$ (This observation appears to first have been made by Tarski \cite{Tarski-1953}). 

This leads to the question of whether any extension of Presburger arithmetic by a polynomial-like function in general is undecidable. It is easy to see that, for integer-valued polynomials $p$ of degree greater than $2$, the extension of Presburger arithmetic by that polynomial $\theory{\NN; +, p}$ is undecidable. This is because we can define multiplication in the extension in much the same way as we could with the squaring function. B\`es \cite{Bes-2001} surveys further decidable and undecidable extensions of Presburger arithmetic, some of which are shown to be undecidable by defining multiplication in Presburger arithmetic.

\subsection{$\checkmark$New Results}
\label{sec:new-results}
Hardy field functions (with polynomially bounded growth), discussed in detail in Section \ref{sec:hardy}, are a far-reaching generalisation of polynomial sequences which is often studied in combinatorial number theory and ergodic theory, see e.g.\ \cite{Boshernitzan-1981}, \cite{Boshernitzan-1994}, \cite{Frantzikinakis-2009} and \cite{Richter-2022}. They are particularly well-behaved from the point of view of analysis, which often allows one to adapt arguments originally applied to polynomials to this wider class of functions. 

While we postpone proper definition of Hardy field functions, for now we point out that they contain all \emph{logarithmico-exponential functions} --- that is, the functions that can be built up from the basic arithmetical operations over the reals (addition, multiplication, subtraction, division), exponentiation and the taking of logarithms. Thus, for instance
\begin{eqnarray*}
f(x) = x^2 + 3x + 2,&g(x) = \sqrt{x}\log x + \frac{1}{\sqrt[3]{x}}, \text{ and}&h(x) = 2^{x^{3/5} + \log x} - 2
\end{eqnarray*}
are all Hardy field functions.

Since we are interested in extensions of the Presburger arithmetic --- which concerns the integers --- and since Hardy field functions generally take real values we need a way to construct integer-valued variants of Hardy field functions. For this, we recall that for $x \in \RR$ the operator $\nint{x}$ denotes the best integer approximation of $x$, meaning that $\nint{x}-1/2 \leq x < \nint{x}+1/2$. We will also occasionally use the floor and the ceiling functions, $\floor{x}$ and $\ceil{x}$, respectively. For a function $f$ we let $\nint{f}$ denote the function that rounds $f$ to the nearest integer (i.e.\ $\nint{f}(x) = \nint{f(x)} = \ceil{f(x)-(1/2)}$). Another technical issue we need to deal with is that Hardy field functions are generally defined only on an interval $[x_0,\infty)$ rather than on all of $\RR$. For the sake of simplicity, by convention, we extend any Hardy field function to $\RR$ by assigning the value $0$ outside of the domain of definition. We are ultimately interested in evaluating the functions under consideration on large integers, so this issue does not affect the reasoning.

We now consider a question raised in \cite{Konieczny-2024}: Given a Hardy field function whose rate of growth is polynomial and faster than linear, is the first-order theory $\theory{\ZZ; <, +, \nint{f}}$ (Note that $\theory{\ZZ; <, +}$ is equivalent to $\theory{\NN; +}$) decidable? In this paper we answer this question negatively. Formally, we prove the following theorem:
\begin{alphatheorem}
\label{thm:main}
Let $f: [x_0,\infty) \to \RR$ be a Hardy field function such that
\begin{eqnarray}
\lim_{x \to \infty} \frac{f(x)}{x} = \infty &\text{ and } & \lim_{x \to \infty} \frac{f(x)}{x^{d}} = 0
\end{eqnarray}
for an integer $d \geq 2$. Then the first-order theory of $\theory{\ZZ; <, +, \nint{f}}$ is undecidable.
\end{alphatheorem}

Similarly, we can also deal with Hardy field functions whose rate of growth is slower than linear but not too slow.
\begin{alphatheorem}
\label{thm:other}
Let $f: [x_0,\infty) \to \RR$ be a Hardy field function such that
\begin{eqnarray}
\lim_{x \to \infty} \frac{f(x)}{x^{\e}} = \infty &\text{ and } & \lim_{x \to \infty} \frac{f(x)}{x} = 0
\end{eqnarray}
for some $\e > 0$. Then the first-order theory of $\theory{\ZZ; <, +, \nint{f}}$ is undecidable.
\end{alphatheorem}

Broadly, then, the extension of Presburger arithmetic by most polynomial-like Hardy field functions yields an undecidable theory, as expected.

\subsection{$\checkmark$Proof Outline}

We prove Theorem \ref{thm:main} in two parts; we first treat the \emph{super-linear} case, where the Hardy field function $f$ grows faster than a quadratic polynomial, and the \emph{near-linear} case, where $f$ grows faster than a linear polynomial but slower than a quadratic polynomial. Building on these results, we then prove Theorem \ref{thm:other}, or the \emph{sub-linear} case, where $f$ grows slower than a linear function, separately.

In the super-linear case, we define multiplication using the fact that, over specific and arbitrarily large intervals, $f$ can be closely approximated by a Taylor polynomial. By differentiating such polynomials, we can define multiplication between arbitrarily large numbers, and so define multiplication generally.

In the near-linear case, we use the near-linear growth of $f$ to define an interval
where $f$ looks like a linear function. We show that there are sufficiently many of these intervals to define multiplication over a restricted domain. By taking difference sets, we then extend this restricted domain to the whole of $\ZZ \times \ZZ$. This defines multiplication generally.

In the sub-linear case, we use the fact that the inverse of $f$ can be approximated in our theory, as well as the fact that the inverse of $f$ grows faster than $x$, to define a near-linear or super-linear rounded Hardy field function. We then apply Theorem \ref{thm:main}, and undecidability follows.

\subsection{$\checkmark$Future directions}
\label{subsec:future-directions}

\subsubsection{Weakening of conditions}

In proving Theorem \ref{thm:main}, we use only a few properties that Hardy field functions have, namely that they are $d$-times differentiable and have well-behaved Taylor approximations. It would be interesting to see which broader classes of functions also satisfy these requirements, as well as seeing how these requirements may be weakened further (particularly with regards to the equidistribution results that are required in the proof of Theorem \ref{thm:main}).

\subsubsection{Extension to non-polynomial functions}

The proofs above also focus on the case where $f$ grows only polynomially quickly, but it might be interesting to see what non-polynomial Hardy field functions provide undecidable extensions of Presburger arithmetic. It seems that functions that grow exponentially or faster, and their inverses, would provide decidable extensions of Presburger arithmetic by ideas similar to Sem\"enov's \cite{Semenov-1979,Semenov-1983}. But functions that grow between polynomially and exponentially fast seem more interesting.

In particular, it seems that using the function $f(x) = 2^{\sqrt{x}}$ we can define a function with polynomial growth. Taking $\inv{f}(f'(x)) - x$, we get a function that looks similar to $\sqrt{x} \log x$, and by the results of this paper rounding this function and extending Presburger arithmetic by it leads to an undecidable theory. This then raises the question of how far this idea can be generalised:

\begin{question}
Given a Hardy field function $f: [x_0,\infty) \to \RR$ such that
\begin{eqnarray}
\lim_{x \to \infty} \frac{f(x)}{x^n} > \infty &\text{ and } & \lim_{x \to \infty} \frac{f(x)}{m^x} = 0
\end{eqnarray}
for all integers $n$ and $m > 1$, is the first-order theory of $\theory{\ZZ; <, +, \nint{f}}$ decidable?
\end{question}

\subsubsection{Generalisation to other theories}

Given that the theory of Skolem arithmetic, namely $\theory{\NN; \times}$, is decidable \cite{Bes-2001}, it would be interesting to see whether the similar theory of $\theory{\NN; \times, \nint{f}}$ is decidable as well. This relates to a question raised by Korec \cite{Korec-2001} as to whether the theory $\theory{\NN; \times, X}$ is decidable, where $X$ is the image of some polynomial function. In particular, it would be interesting to see whether Hardy field functions can in general be used to define addition as well as multiplication, and so used to show undecidability of the above theory. This leads us to the following question:

\begin{question}
Let $f \colon \RR \to \RR$ be a Hardy field function such that
\begin{eqnarray}
\lim_{x \to \infty} \frac{f(x)}{x^d} = \infty &\text{ and } & \lim_{x \to \infty} \frac{f(x)}{x^{d+1}} = 0
\end{eqnarray}
for an integer $d > 1$. Is the first-order theory of $\theory{\ZZ; \times, \nint{f}}$ decidable?
\end{question}

A further question along these lines would be whether just the theory $\theory{\NN; \nint{f}}$ is decidable as well.

\subsection{$\checkmark$Notation}
\label{sec:notation}
We let $\NN$ be the set of nonnegative integers $\{0,1,2, \dots\}$.

For $x \in \RR$, we define the integer part $\ip{x}$ of $x$ to be $\max \{n \in \ZZ \mid n \leq x\}$. We then define the ceiling of $x$ to be $\ceil{x} = \ip{-x}$ and the nearest integer of $x$ to be $\nint{x} = \ceil{x - 1/2}$. When applying the nearest integer function to a function generally, we write $\nint{f}(x)$ for $\nint{f(x)}$. We also write the circle norm of $x$ as $\fpa{x} = \min \set{\abs{x-n}}{n \in \ZZ}$ to denote the distance of $x$ to the nearest integer.

\subsection*{$\checkmark$Acknowledgements}

The second named author is supported by UKRI Fellowship EP/X033813/1. For the purpose of open access, the authors have applied a Creative Commons Attribution (CC BY) licence to any Author Accepted Manuscript version arising.

\section{$\checkmark$Hardy fields}
\label{sec:hardy}

We define Hardy fields as follows. Let $B$ be a set of equivalence classes of continuous real-valued functions in one variable, where we say that two such functions $f$ and $g$ are equivalent when they eventually agree i.e.\ when there exists some $x_0$ such that, for all $x > x_0$, we have $f(x) = g(x)$. Such equivalence classes are also called \emph{germs}. The set $B$ naturally gives rise to a ring $(B, +, \times)$. We say that a Hardy field is a subfield of the ring $(B, +, \times)$ that is closed under differentiation. A Hardy field function is a function that belongs to a Hardy field, or more precisely whose germ at $\infty$ belongs to the union of all Hardy fields.

The union of all Hardy fields is large enough to include a variety of interesting functions. This includes, as mentioned in Section \ref{sec:new-results}, the class of logarithmico-exponential functions built up from real polynomials, exponentiation, and the taking of logarithms. This gives us both a natural class of examples of Hardy field functions, as well as a natural class of functions to compare general Hardy field functions to.

Hardy field functions exhibit many properties that make them well suited to analytic arguments. As a first instance of this principle, because of the fact that every Hardy field function is asymptotically comparable to $0$, we have the following standard fact. 

\begin{lemma} \label{lem:eventually-increasing-or-decreasing}
	Let $f \colon \RR \to \RR$ be a Hardy field function.
	Then $f$ is either eventually positive, eventually negative, or eventually zero. Likewise, $f$ is either eventually increasing, eventually decreasing, or eventually constant.
\end{lemma}

As a consequence, it makes sense to compare rates of growth of different Hardy field functions. Given two eventually positive functions $f$ and $g$ belonging to the same Hardy field, we write:
\begin{itemize}
\item $f(x) \ll g(x)$ if $\lim_{x \to \infty} f(x)/g(x) < \infty$;
\item $f(x) \prec g(x)$ if $\lim_{x \to \infty} f(x)/g(x) = 0$;
\item $f(x) \sim g(x)$ if $\lim_{x \to \infty} f(x)/g(x) \in (0,\infty)$.
\end{itemize}
Note that $f(x) \ll g(x)$ holds if and only if either $f(x) \prec g(x)$ or $f(x) \sim g(x)$.

In particular, throughout we consider Hardy field functions $f \colon [x_0, \infty) \to \RR$; by this we mean that $f$ behaves as usual on the half-line $[x_0, \infty)$, and that $f(x) = 0$ for all $x < x_0$. We pick $x_0$ such that $f$ is either strictly increasing or strictly decreasing over the interval $[x_0, \infty)$, which by Lemma \ref{lem:eventually-increasing-or-decreasing} we are licenced to do. This does not affect the correctness of our arguments when applied to Hardy field functions generally, but merely simplifies proofs.

\begin{lemma}
	Let $f,g \colon [x_0, \infty) \to \RR$ be eventually positive functions belonging to the same Hardy field. Then either $f(x) \prec g(x)$ holds, $f(x) \sim g(x)$ holds, or $f(x) \succ g(x)$ holds. Further, exactly one of the above holds.
\end{lemma}

One convenient feature of Hardy field functions is that their derivatives have a rate of growth that is easy to describe, as shown in the following result.

\begin{lemma}[{\cite[Lem.\ 2.1]{Frantzikinakis-2009}}]\label{lem:hardy:derivative}
Let $f \colon [x_0, \infty) \to \RR$ be a Hardy field function that satisfies $t^{-d} \ll \abs{f(t)} \ll t^d$ for some $d \in \NN$. Then $\abs{f'(t)} \ll \abs{f(t)}/t$.
\end{lemma}

Another convenient feature of Hardy field functions is that they can be accurately approximated by their Taylor expansions. Given a function $f \colon [x_0,\infty) \to \RR$ that has sufficiently many derivatives, for $x \geq x_0$ and $y \geq 0$, we can always consider the length-$\ell$ Taylor expansion 
\begin{align*}\label{eq:hardy:taylor-decomp}
	f(x+y) = P_{x,\ell}(y) + R_{x,\ell}(y),
\end{align*}
where $P_{x,\ell}$ is the Taylor polynomial
\begin{align*}
	P_{x,\ell}(y) = f(x) + y f'(x) + \cdots + \frac{y^{\ell-1}}{(\ell-1)!} f^{(\ell-1)}(x),
\end{align*}
and $R_{x,\ell}$ is the remainder term (which we can consider to be defined simply as $R_{x,\ell}(y) = f(x+y) - P_{x,\ell}(y)$). Since Hardy field functions do have sufficiently many derivatives, we can consider such Taylor expansions of Hardy field functions. We will use the following standard estimate (similar results can be found e.g.\ in \cite{Frantzikinakis-2009} or \cite[Prop.\ 2.9]{KoniecznyMullner-2024}).

\begin{proposition}\label{prop:hardy:taylor}
	Let $f \colon [x_0, \infty) \to \RR$ be a Hardy field function that satisfies $t^{d-1} \ll \abs{f(t)} \prec t^d$ for some $d \in \NN$. Then we have
\begin{align*}
	R_{x,d}(y) \ll y^{d} f(x)/x^{d+1}
\end{align*}
for sufficiently large $x$ and all $0 \leq y \leq x$. (The constant implicit in the asymptotic notation depends only on $f$).
\end{proposition}
\begin{proof}
Recall that for any $x \geq x_0$ and $y \geq 0$ there exists $z \in [0,y]$ such that
	\begin{align*}
	R_{x,d}(y) = \frac{y^{d}}{d!} f^{(d)}(x+z).
	\end{align*}
	Iterating Lemma \ref{lem:hardy:derivative} $d$ times, we see that $f^{(d)}(t) \ll f(t)/t^d \to 0$ as $t \to \infty$. Since $f^{(d)}$ is a Hardy field function and is eventually monotone, we have
	\begin{equation*}
	\abs{R_{x,d}(y)} \leq \frac{y^{d}}{d!} \abs{ f^{(d)}(x)} \ll y^d f(x)/x^d. \qedhere
	\end{equation*}
\end{proof}

We will also need the following results on distribution of Hardy field sequences. (We point out that \cite{Boshernitzan-1994} in fact provides if and only if statements, but we will only be interested in one direction).

\begin{theorem}[{\cite[Thm.\ 1.9]{Boshernitzan-1994}}]\label{thm:hardy:equidist}
Let $f_1,f_2,\dots,f_k$ be functions belonging to the same Hardy field. 

Suppose that for each $n_1,n_2,\dots,n_k \in \ZZ$, not all zero, and for every polynomial $p(x) \in \QQ[x]$, letting 
\[
g(x) = n_1 f_1(x) + n_2 f_2(x) + \dots + n_k f_k(x) - p(x),\]
we have
$\lim_{x \to \infty} {\abs{g(x)}}= \infty$. Then the sequence 
\[
\bra{f_1(n) \bmod 1,f_2(n) \bmod 1,\dots, f_k(n) \bmod 1}_{n=0}^\infty\]
is dense in $\RR^k/\ZZ^k$.
\end{theorem}
\begin{remark}
Note that, on the stronger condition that $\lim_{x \to \infty} \abs{g(x)}/\log x = \infty$, the sequence $\bra{f_1(n) \bmod 1,f_2(n) \bmod 1,\dots, f_k(n) \bmod 1}_{n=0}^\infty$ is also uniformly distributed in $\RR^k/\ZZ^k$ \cite[Thm.\ 1.8]{Boshernitzan-1994}.
\end{remark}

\section{$\checkmark$Proof of the super-linear case}
In this section we will prove Theorem \ref{thm:main} in the case where when $f(x) \gg x^2$. Broadly speaking, we will prove that all such Hardy field functions have what will be called the \emph{$P_d$ property}, which expresses that a function looks like a polynomial over arbitrarily long intervals. From that, we will introduce the $P_d^{\ZZ}$ property, which expresses that $f$ looks like an \emph{integer-valued} polynomial over arbitrarily long intervals. We will then show that if $f$ has the $P_d^{\ZZ}$ property, then the first-order theory $\theory{\ZZ; <, +, \nint{f}}$ is undecidable. A key step of the argument is differentiating the polynomial given by the $P_d^{\ZZ}$ property. 

To begin with, we conduct the proof under the additional assumption that $x^{d-1} \prec f(x) \prec x^{d}$ for some $d \geq 3$. In particular, we initially exclude the case where $f(x) \sim x^{d-1}$, which we cover in a separate Subsection \ref{sec:super-linear-generalisation}. The proof in the latter case follows along broadly similar lines. However, failure of certain equidistribution results forces us to introduce some new ideas.

\subsection{$\checkmark$The $P_d$ property}
\label{sec:pd}
We first introduce the \emph{$P_d$ property}, which holds of a function when that function can be approximated arbitrarily closely, and over arbitrarily long intervals, by some polynomial of degree less than $d$. Formally, we will say that a function $f \colon [x_0,\infty) \to \RR$ has property $P_d$ for some $d \in \NN$ if for each $M \in \NN$ and $\e > 0$ there exists some $N \in \NN$ and a degree-$(d-1)$ polynomial $p$ such that for all $0 \leq m < M$ we have $\abs{f(N+m) - p(m)} < \e$. A convenient feature of Hardy field functions with polynomial growth is that they enjoy the property $P_d$, as shown in the following lemma.

\begin{lemma}
\label{lem:pd}
Let $f$ be a Hardy field function such that { $x^{d-1} \ll f(x) \prec x^{d}$} for some $d \in \NN$. Then $f$ has the $P_d$ property. 
\end{lemma}
\begin{proof}
Pick any $M \in \NN$ and $\e > 0$, and let $N$ be a large integer, to be determined in the course of the argument. Recall that we can expand $f(N+m)$ as { $P_{N,d}(m) + R_{N,d}(m)$}, where { $P_{N,d}(m)$} is the degree-{$(d-1)$} Taylor polynomial and { $R_{N,d}(m)$} is the corresponding remainder term. Assuming, as we may, that $N > M$, for $0 \leq m \leq M$ by Proposition \ref{prop:hardy:taylor} we have $\abs{R_{N,{ d}}(m)} \ll M^{{ d}} f(N)/N^{{ d}}$. Since $f(N)/N^{{ d}} \to 0$ as $N \to \infty$, picking sufficiently large $N$ we can ensure that $\abs{R_{N,{ d}}(m)} \leq \e$, as needed.
\end{proof}

\newcommand{\Pint}{P^{\ZZ}}

\subsection{$\checkmark$The $\Pint_d$ property}

Recall that we are ultimately interested not in a Hardy field function $f$ but rather in its integer-valued rounding $\nint{f}$. Like in the previous section, let us suppose for a moment that on some interval $[N,N+M)$ the Hardy field function $f$ is closely approximated by a degree-${ (d-1)}$ polynomial $p$, in the sense that for all $0 \leq m < M$ we have $\abs{f(N+m)-p(m)} < \e$ for some small constant $\e > 0$. In this situation, it is not necessarily the case that $\nint{f}$ is closely approximated by $\nint{p}$. Indeed, if for some $0 \leq m < M$ we have $f(N+m) > k+(1/2) > p(m)$ with $k \in \ZZ$ then $\nint{f(N+m)} = k+1 \neq k = \nint{p(m)}$. As a first step, we would like to avoid this behaviour, which is most easily accomplished by requiring that $\fpa{f(N+m)}$ and $\fpa{p(m)}$ are both small. Secondly, we note that (in the regime where $M$ is much larger than $d$) the only way for $\fpa{p(m)}$ to be small for all $0 \leq m < M$ is if $p$ is closely approximated by an integer-valued polynomial, i.e.\ a polynomial $q$ such that $q(\ZZ) \subset \ZZ$ (since we only use this statement as a source of intuition, we leave it admittedly vague). This motivates us to introduce a property $\Pint_d$, which is an analogue of the property $P_d$ discussed earlier. We will say that a function $f \colon [x_0,\infty) \to \RR$ has property $\Pint_d$ for some $d \in \NN$ if for each $M \in \NN$ and $\e > 0$ there exists $N \in \NN$ and a degree-${ (d-1)}$ \emph{integer-valued} polynomial $p$ such that for all $0 \leq m < M$ we have $\abs{f(N+m)-p(m)} < \e$. Above, we require that the polynomial $p$ should have degree \emph{exactly} ${ d-1}$, as opposed to \emph{at most} ${ d-1}$; this requirement will play an important role in later considerations.

\begin{lemma}\label{lem:pint}
Let $f$ be a Hardy field function such that $x^{{ d-1}} \prec f(x) \prec x^{{ d}}$ for some $d \in \NN$. Then $f$ has the $\Pint_d$ property.
\end{lemma}
\begin{proof}
	Pick any $M \in \NN$ and $\e > 0$. Recall from the proof of Lemma \ref{lem:pd} that for sufficiently large $N$ we can accurately approximate $f$ on $[N,N+M)$ using the Taylor expansion $P_{N, { d}}$, meaning that $\abs{f(N+m)-P_{N,{ d}}(m)} < \e/2$ for all $0 \leq m < M$. Recall also that the coefficients of $P_{N,{ d}}$ are given by
	\begin{align*}
		P_{N,{ d}}(m) = \sum_{k=0}^{{ d-1}} \frac{f^{(k)}(N)}{k!} m^k.
	\end{align*}
Consider the polynomial $p$ obtained from $P_{N,{ d}}$ by applying rounding to each coefficient:	
	\begin{align*}
		p(m) = \sum_{k=0}^{{ d-1}} \nint{ \frac{f^{(k)}(N)}{k!} } m^k.
	\end{align*}	
	Since $p$ has integer coefficients, it is clearly integer-valued. Our plan is to show that, for a judicious choice of $N$, the circle norms of coefficients of $P_{N,{ d}}$ are small and consequently $f$ is closely approximated by $p$. 
	
	We will apply Theorem \ref{thm:hardy:equidist} to the functions $f, f', \dots, f^{(d-1)}/(d-1)$. For each $0 \leq k \leq d$ we have $f^{(k)}(x) \sim f(x)/x^k$. As a consequence, for any non-trivial linear combination $h(x) = c_0 f(x) + c_1 f'(x) + \dots + c_{d-1} f^{(d-1)}(x)/(d-1)!$ with $c_0,c_1,\dots,c_{d-1} \in \ZZ$ we have $h(x) \sim f(x)/x^k$, where $\ell$ is the first index with $c_\ell \neq 0$. Thus, for any polynomial $p$ with degree $e$ we have $\abs{ h(x) - p(x)} \sim f(x)/x^\ell$ if $d-\ell > e$ or $\abs{ h(x) - p(x)} \sim x^e$ otherwise (that is, if $d-\ell \leq e$). In either case, we have $\abs{ h(x) - p(x)} \succ 1$, as needed. We conclude that the sequence $\bra{f(N), f'(N), \dots, f^{(d-1)}(N)/(d-1)!}_{N=0}^\infty$ is dense modulo $1$, and in particular it includes points with arbitrarily small circle norms. Hence, we can find $N$ such that for all $0 \leq k < d$ we have:
\begin{align*}
		\fpa{\frac{f^{(k)}(N)}{k!} } \leq \frac{\e}{2{ d}M^k}.
\end{align*}	  
Therefore, we have the estimate
\begin{align*}
	\abs{P_{N,{ d}}(m) - p(m)} \leq \sum_{k=0}^{{ d-1}} \fpa{\frac{f^{(k)}(N)}{k!} } M^k \leq \e/2.
\end{align*} 
Combining this with the previously mentioned estimate on $\abs{f(N+m)-P_{N,{ d}}(m)} $ we conclude that $\abs{f(N+m)-p(m)} < \e$, as needed. Finally, increasing $N$ if necessary, we may assume that the leading coefficient of $p$, i.e. $\nint{f^{(d-1)}(N)/(d-1)!}$, is non-zero.
\end{proof}

\subsection{$\checkmark$Discrete derivatives}
\label{subsec:discrete-derivatives}

Using Lemma \ref{lem:pint}, for a Hardy field function $f$ satisfying the assumptions of Theorem \ref{thm:main}, we can find arbitrarily long intervals $[N,N+M)$ where $\nint{f}$ agrees some integer-valued polynomial $p$. The goal of the next two subsections is to prove the following lemma, asserting that this property implies that the theory $\theory{\ZZ;<,+,\nint{f}}$ is undecidable.

\begin{lemma}
\label{lem:pd-to-undecidable}
Let $f \colon [x_0,\infty) \to \RR$ be a function that satisfies the $\Pint_d$ property for some $d \geq 3$. Then the theory $\theory{\ZZ; <, +, \nint{f}}$ is undecidable.
\end{lemma}

A key ingredient of the proof of Lemma \ref{lem:pd-to-undecidable} is the notion of differentiation for sequences indexed by integers, which will ultimately help us define multiplication. To this end, we introduce the \emph{discrete derivative} and the \emph{symmetric discrete derivative} of a function:

\begin{definition}

Given a function $f \colon \ZZ \to \RR$ and an integer $m$, the \emph{discrete derivative} $\DeltaT_m f \colon \ZZ \to \RR$ is given by
\begin{equation*}
\DeltaT_m f(n) = f(n+m) - f(n)
\end{equation*}
and the \emph{symmetric discrete derivative} $\DeltaS{}_m f \colon \ZZ \to \RR$ is given by
\begin{equation*}
\DeltaS{}_m f(n) = \DeltaT_m \DeltaT_n f(0) = f(n+m) - f(n) - f(m) + f(0).
\end{equation*}
Of course, the discrete derivative of an integer-valued function is again integer-valued. We also point out that the discrete derivative operators commute: $\Delta_m \Delta_n = \Delta_n \Delta_m$. The symmetric derivative, as the name suggests, is a symmetric function of the arguments: $\DeltaS{}_m f(n) = \DeltaS{}_n f(m)$.
For reasons of symmetry, given an integer $r \geq 1$ and integers $n_0, n_1, \dots, n_r$ we write
\begin{eqnarray*}
\DeltaS{}^r f(n_0, n_1, \dots, n_r) & \text{ for } & \DeltaS{}_{n_r} \DeltaS{}_{n_{r-1}} \dots \DeltaS_{n_1} f(n_0).
\end{eqnarray*}
\end{definition}

A particularly useful feature of the discrete derivative (symmetric or otherwise) is that its application to a polynomial yields a polynomial of degree one less. As a consequence, its repeated application only leaves the leading term of the function to be considered. We make this observation concrete in the following result.

\begin{lemma}
\label{lem:symmetric-discrete-degree}
Let $m$ be an integer and let $p$ be a polynomial of degree $r$ with leading coefficient $a_r$. 
\begin{enumerate}
\item\label{lem:symmetric-discrete-degree:i} If $r \geq 1$ then $\Delta_m p$ is a polynomial with degree $r-1$ and leading coefficient $ra_rm$. If $r = 0$ then $\Delta_m p = 0$.
\item\label{lem:symmetric-discrete-degree:ii} If $r \geq 2$ then $\DeltaS_m p$ is a polynomial with degree $r-1$ and leading coefficient $ra_rm$. If $r = 0$ or $r = 1$ then $\DeltaS_m p = 0$.
\end{enumerate}
\end{lemma}
\begin{proof}
It is straightforward to verify by an explicit computation that item \ref{lem:symmetric-discrete-degree:i} holds for $r = 0$ and item \ref{lem:symmetric-discrete-degree:ii} holds for $r = 0,1$. For $r \geq 2$, $\Delta_m p$ and $\DeltaS_m p$ differ only by a constant (equal to $p(m)-p(0)$) so it suffices to prove item \ref{lem:symmetric-discrete-degree:i}.

We proceed by induction on $r$, the case $r = 0$ already having been considered. Thus, we may assume that $r \geq 1$ and the claim has already been proved for all $r' < r$. 

We may write $p(x) = a_r x^r + \tilde p(x)$ where $\deg \tilde p < r$. By the inductive assumption, $\Delta_m \tilde p$ is a polynomial of degree strictly less than $r-1$ (or identically zero if $r = 1$), so it suffices to deal with the leading term. We can explicitly compute that
\begin{align*}
	a_r(x+m)^r - a_r x^r = \sum_{k=1}^{r} \binom{r}{k} a_r m^k x^{r-k}, 
\end{align*}
where the right side is a degree-$(r-1)$ polynomial with leading coefficient $r a_r m$, as needed. 

\end{proof}

For brevity, we write $\DeltaS{}^r f(a,b)$ for $\DeltaS{}^r f(a,b, 1, 1, \dots, 1)$ whenever $r \geq 1$. As an application of Lemma \ref{lem:symmetric-discrete-degree} we almost immediately get the following formula.

\begin{lemma}\label{lem:symmetric-discrete-degree^d}
Let $p$ be a polynomial with degree $r$ and leading coefficient $a_r$. Then $\DeltaS^{r-1} p(n,m) = r! a_r nm$.
\end{lemma}
\begin{proof}
Iterating  Lemma \ref{lem:symmetric-discrete-degree}  we see that $\DeltaS^{r-1} p(n,m)$ is a polynomial function of $n$ with degree $1$ and leading coefficient $r! a_r m$. To see that all the remaining coefficients are zero, it is enough to recall that $\DeltaS^{r-1} p(n,m)$ is a symmetric function of $n$ and $m$, and that $\DeltaS^{r-1} p(0,0) = 0$.
\end{proof}

In the direction opposite to Lemma \ref{lem:symmetric-discrete-degree}, we have the following characterisation of polynomials in terms of discrete derivatives. It is a standard observation for instance following from discussion in \cite[Section 2.6]{GKP-1990}; we include proof for completeness.

\begin{lemma}\label{lem:symmetric-discrete-char-poly}	
	Let $r \geq 0$ and let $f \colon [N,N+M) \to \RR$ be a sequence such that $\Delta^{r+1}_1 f(n) = 0$ for $N \leq n < N+ M-r-1$. Then $f$ coincides with a polynomial of degree at most $r$.
\end{lemma}
\begin{proof}
	For any polynomial $p$ of degree at most $r$ we have $\Delta^{r+1} p = 0$, so we may freely replace $f$ with $f-p$. Applying Lagrange interpolation, we may thus assume that $f(N) = f(N+1) = \dots = f(N+r) = 0$. Since $\Delta^{r+1}_1 f(n) = 0$ for all $n$ where it is defined, we see that if we have $f(n) = f(n+1) = \dots = f(n+r) = 0$ for some $N \leq n < N+M-r-1$ then also $f(n+r+1) = 0$. Reasoning by induction with respect to $n$ we thus conclude that $f(m) = 0$ for all $N \leq m < N+M$. In particular, $f$ is a polynomial of degree at most $r$
\end{proof}

\subsection{$\checkmark$Emulating multiplication}
Using the results obtained in Section \ref{subsec:discrete-derivatives}, we next show how to use property $\Pint_d$ to emulate multiplication. Recall that $\Pint_d$ implies that we can find arbitrarily long intervals $[N,N+M)$ where $\nint{f}$ agrees with an integer-valued polynomial $p$ (in the sense that $\nint{f}(N+m) = p(m)$). What is more, we can use Lemma \ref{lem:symmetric-discrete-char-poly} to detect intervals with the property mentioned above. Given such an interval, for $1 \leq a,b,c \leq M$ we can use Lemma \ref{lem:symmetric-discrete-degree^d} to express the property that $ab=c$ as $\DeltaS^{d-2} p(a,b) = \DeltaS^{d-2} p(c,1)$. We now put this plan into practice.

\begin{lemma}\label{lem:polynomial-sequence-agreement}
	Let $d \geq 2$ and let $\pi_d(N,M)$ be the sentence given by
	\begin{align*}
	\exists c \neq 0 \ \forall 0 \leq m < M-d \  \Delta^{d-1}_1 \nint{f}(N+m) = c.   
	\end{align*}
	Then $\pi_d(N,M)$ holds if and only if there exists a polynomial $p$ with degree exactly $d-1$ such that $\nint{f}(N+m) = p(m)$ for all $0 \leq m < M$.
\end{lemma}
\begin{proof}
	Suppose that $\pi(N,M)$ holds. Applying $\Delta_1$ once more, we conclude from Lemma \ref{lem:symmetric-discrete-char-poly} that $\nint{f}$ agrees with a polynomial $p$ of degree at most $d-1$ on $[N,N+m)$. If $p$ had degree strictly less than $d-1$ then, by repeated application of Lemma \ref{lem:symmetric-discrete-degree}, we would have $\Delta^{d-1}_1 \nint{f}(N+m) = 0 \neq c$, contradicting $\pi_d(N,M)$. Thus, $p$ has degree exactly $d-1$, as needed.
	
	Suppose now that there exists a polynomial $p$ with degree exactly $d-1$ such that $\nint{f}(N+m) = p(m)$ for all $0 \leq m < M$. Then, by repeated application of Lemma \ref{lem:symmetric-discrete-degree}, we have $\Delta_1^{d-1} \nint{f}(N+m) = (d-1)! a_{d-1}$, where $a_{d-1}$ is the leading coefficient of $p$. Thus, setting $c = (d-1)! a_{d-1}$ we see that $\pi_d(N,M)$ holds.
\end{proof}

\begin{lemma}\label{lem:multiplication-definition}
	Let $d \geq 3$ and let $\mu_d(n,m,q)$ be the sentence given by
	\begin{align*}
	\exists M > \max(n,m,q) \ \exists N \ \pi_d(N,M) \wedge 
	\Delta_1^{d-3} \Delta_n \Delta_m \nint{f}(N) =
	\Delta_1^{d-2} \Delta_q \nint{f}(N).
	\end{align*}
	Assume that $f$ enjoys property $\Pint_d$ and $n,m,q \geq 1$. Then, for $n,m,q \in \NN$ we have that $\mu_d(n,m,q)$ holds if and only if $q = nm$.
\end{lemma}
\begin{proof}
Suppose that $\mu(n,m,q)$ holds. Pick admissible $M$ and $N$. Since, in particular, $\pi_d(N,M)$ holds, we know from Lemma \ref{lem:polynomial-sequence-agreement} that there exists a polynomial $p$ of degree $d-1$ such that $\nint{f}(N+m) = p(m)$ for $0 \leq m < M$. The equality between the discrete derivatives in the definition of $\mu(n,m,q)$ can now more simply be expressed as $\DeltaS^{d-2}p(n,m) = \DeltaS^{d-2}p(q,1)$. By Lemma \ref{lem:symmetric-discrete-degree^d}, this implies that $nm = q$, as needed. 

Next, suppose that $q = nm$. Pick $M = q+1$. By $\Pint_d$, we can find $N$ and a polynomial $p$ of degree exactly $d-1$ such that $f(N+m) = p(m)$ for $0 \leq m < M$. By Lemma \ref{lem:polynomial-sequence-agreement}, $\pi_d(M,N)$ holds. By Lemma \ref{lem:symmetric-discrete-degree}, we have 
$\DeltaS^{d-2}p(n,m) = (d-1)! a_{d-1} nm = (d-1)! a_{d-1} q = \DeltaS^{d-2}p(q,1)$. Hence $\mu(n,m,q)$ holds, as needed.
\end{proof}

\begin{proof}[Proof of Lemma \ref{lem:pd-to-undecidable}]
It follows from Lemma \ref{lem:multiplication-definition} that multiplication on $\NN$ is definable in $\theory{\ZZ; <, +, \nint{f}}$, and extending it to $\ZZ$ is immediate. Since $\theory{\ZZ; <, +, \times}$ is undecidable, so is $\theory{\ZZ; <, +, \nint{f}}$.
\end{proof}

Combining Lemmas \ref{lem:pint} and \ref{lem:pd-to-undecidable}, we conclude that for each $d \geq 3$ and each Hardy field function $f$ with $x^{d-1} \prec f(x) \prec x^d$, the theory $\theory{\ZZ;<,+,\nint{f}}$ is undecidable. This completes the proof of the first case of Theorem \ref{thm:main}.

\subsection{$\checkmark$Generalisation to exactly polynomial growth}
\label{sec:super-linear-generalisation}
Finally, we show how the argument in the earlier sections to sequences with exactly polynomial growth. Recall that previously we considered Hardy a field function $f$ with $x^{d-1} \prec f(x) \prec x^d$ for some $d \geq 3$. Presently, we will instead assume that $f(x) \sim x^{d-1}$ (note that we make this choice instead of the more natural $f(x) \sim x^{d}$ for the sake of consistency with earlier considerations). 

Significant part of the previously presented argument goes through without any change. Indeed, it remains the case that $f$ enjoys the $P_{d}$ property and can be accurately approximated by the Taylor polynomial $P_{N,d}$. Unfortunately, we are not able to establish the $\Pint_d$ property. When we try to repeat the previous argument, the $d$-tuple formed by the coefficients of $P_{N,d}$, namely
\begin{equation*}
(f(N), f'(N), f''(N)/2, \dots, f^{(d-1)}(N)/(d-1)!),
\end{equation*}
is not equidistributed modulo $1$. Indeed, the top coefficient $f^{(d-1)}(N)/(d-1)!$ can easily be shown to converge to a constant, and the behaviour of the remaining coefficients is potentially more complicated.

Because of the aforementioned limitation, we adapt the remainder of the argument to use property $P_d$ instead of $\Pint_d$. This introduces some complications since we are forced to apply discrete derivatives to functions of the form $\nint{p}$, where $p$ is a polynomial, rather than to polynomials. Fortunately, the following lemma allows us to control errors arising from the rounding operation; we will apply it to $h(n) = p(n) - \nint{p(n)}$.
\begin{lemma}\label{lem:error-discr-der}
	Let $h \colon \ZZ \to \RR$ be a sequence with $\abs{h(n)} \leq \e$ for all $n \in \ZZ$. Then
	\[ 
	\abs{ \Delta_{m_1} \Delta_{m_2} \dots \Delta_{m_r} h(n)} \leq 2^r \e
	\]
	for all integers $r \geq 1$ and all $n,m_1,m_2,\dots,m_r \in \ZZ$. 
\end{lemma}
\begin{proof}
	For $r = 1$, the result follows from a straightforward computation. For $r > 1$, we use a standard inductive argument.
\end{proof}

Using Lemma \ref{lem:error-discr-der}, we define a coarser variant of multiplication, which we then bootstrap to a definition of multiplication. This finishes the proof of the relevant case of Theorem \ref{thm:main}. We now put the strategy discussed above into practice. Formally, we establish the following result.

\begin{proposition}
\label{prp:integer-exponent}
Let $f \colon [x_0,\infty) \to \RR$ be a Hardy field function where $f(x) \sim x^{d-1}$ for some $d \geq 3$. Then the theory $\theory{\ZZ; <, +, \nint{f}}$ is undecidable.
\end{proposition}
\begin{proof}

Fix an integer $M$ and let $N$ be sufficiently large such that the Hardy field function $f$ is closely approximated by the degree-$(d-1)$ polynomial $P_{N,d}$, in the sense that we have $\abs{f(N+m)-P_{N,d}(m)} < 1/10$ for all $0 \leq m < M$. Recall that the leading coefficient of $P_{N,d}$ is $f^{(d-1)}(N)/(d-1)!$, which converges to some non-zero constant $\a$ as $N \to \infty$. Let $p$ be the polynomial obtained from $P_{N,d}$ by replacing the leading coefficient with $\a$. Picking larger $N$ if necessary, we may freely assume that $\abs{f(N+m)-p(m)} < 1/10$ for all $0 \leq m < M$.

Note that, since $p$ is a polynomial of degree $d-1$, we can apply Lemma \ref{lem:symmetric-discrete-degree^d} to it and get that $\DeltaS^{d-2} p(n,m) = (d-1)! \a nm$. Bearing in mind that we aim to adapt Lemma \ref{lem:multiplication-definition}, we use Lemma \ref{lem:error-discr-der} to approximate
\begin{align*}
\abs{ \Delta_1^{d-3} \Delta_n \Delta_m \nint{f}(N) - (d-1)! \a nm } \leq 2^{d-1},
\end{align*}
provided that $0 \leq n,m < M$. This motivates us to consider, for $0 \leq a,b,c < M$, the quantity $F_N(a,b,c)$ defined as follows:
\begin{equation*}
F_N(a,b,c) = \Delta_1^{d-3} \Delta_a \Delta_b \nint{f}(N) - \Delta_1^{d-2} \Delta_c \nint{f}(N).
\end{equation*}
The estimate obtained above implies that we have
\begin{align*}
	\abs{ F_N(a,b,c) - (d-1!) \a (ab - c) } \leq 2^d.
\end{align*}
Fix an integer $D \geq {100 \cdot 2^d}/{\min(1,(d-1)!\a)}$. The estimate above allows us to prove the following lemma:
\begin{lemma} \label{lem:coarse-multiplication}
Let $0 \leq a,b,c < M$. Then $\abs{F_N(Da,b,Dc)} \leq 2^{d}$ if and only if $ab = c$.
\end{lemma}
\begin{proof}
For the rightwards direction of proof, first suppose that $ab \neq c$. Since $ab - c$ is a non-zero integer, it follows that $\abs{Dab-Dc} \geq D$. Hence $\abs{F_N(Da,b,Dc)} \geq (d-1)!\a D-2^d > 2^d$, as needed.

For the leftwards direction of proof, suppose that $ab = c$. Then $Dab - Dc = 0$ as well. It follows that $\abs{F_N(Da, b, Dc)} \leq 2^{d}$ as required.
\end{proof}

Note that $\abs{F_N(Da, b, Dc)}$ is definable in $\theory{\ZZ; <, +, \nint{f}}$. Thus, Lemma \ref{lem:coarse-multiplication} allows us to define multiplication in $\theory{\ZZ; <, +, \nint{f}}$. More formally, let $\mu'(n,m,q)$ denote the formula 
\begin{equation*}
\forall N_0\ \exists N \geq N_0\ \abs{F_N(Da,b,Dc)} \leq 2^{d}.
\end{equation*}
It follows from the preceding discussion (in particular Lemma \ref{lem:coarse-multiplication}) that $\mu'(n,m,q)$ holds if and only if $ab = c$. Recalling that $\theory{\ZZ; <, +, \times}$ is undecidable and that multiplication is definable in $\theory{\ZZ; <, +, \nint{f}}$, we conclude that $\theory{\ZZ; <, +, \nint{f}}$ is undecidable.	
\end{proof}

\begin{remark}
	One could use the techniques used here in order to establish Theorem \ref{thm:main} also in the case where $x^{d-1} \prec f(x) \prec x^d$. We take the route discussed earlier because we consider it to be more elegant, and because it has the added advantage of allowing us to identify the property $\Pint_d$. 
\end{remark}

\begin{remark}
	We note that previously, we were able to establish undecidability of $\theory{\ZZ; <, +, \nint{f}}$ as the consequence purely of the property $\Pint_d$. This is in contrast with the argument discussed presently, which uses a property strictly stronger than $P_d$, namely that $f$ is closely approximable by a polynomial on the interval $[N,N+M)$ for \emph{all} $N$ that are sufficiently large (as a function of $M$). It would be interesting to determine if the property $P_d$ by itself is sufficient to establish undecidability. The key difficulty to this approach is finding a suitable analogue of Lemma \ref{lem:polynomial-sequence-agreement}.
\end{remark}

\section{$\checkmark$Proof of the near-linear case}

In this section, we prove the second case of Theorem \ref{thm:main} where $f$ grows super-linearly but sub-quadratically i.e.\ when $x \prec f(x)$ and $f(x) \prec x^2$. To this end, we first define multiplication over a limited domain, using the fact that over certain intervals $f$ can be closely approximated by a linear function. We then use difference sets to extend this definition of multiplication to the whole of $\ZZ \times \ZZ$ and thus show $\theory{\ZZ; <, +, \nint{f}}$ undecidable.

\subsection{$\checkmark$Multiplicative intervals}
\label{subsec:multiplicative-intervals}

We first define a property $\lambda(N,M,n)$ which, informally, states that over the interval $[N, N+M)$ the function $\nint{f}$ defines an arithmetic progression with step $n$. Formally, we define $\lambda$ by
\begin{equation*}
\lambda(N,M,n) \defeq (\forall\ 0 \leq m < M)\ \nint{f}(N+m+1) = \nint{f}(N+m)+n.
\end{equation*}

It is routine to show that $\lambda(N,M,n)$ holds if and only if for all $0 \leq m < M$ we have $\nint{f}(N+m) - \nint{f}(N) = mn$. Our next step is to establish sufficient conditions for $\lambda(N,M,n)$ to hold.

\begin{lemma}
\label{lem:lambda-conditions-new}
	Let $f \colon [x_0,\infty) \to \RR$ be a Hardy field function with $x \prec f(x) \prec x^2$. Then there exists $x_1 \geq x_0$ such that $\lambda(N,M,n)$ holds for all $N,M,n \in \NN$ satisfying the following conditions:
\begin{inparaenum}
\item $N \geq x_1$;
\item $\fpa{f(N)} < 1/4$;
\item $\abs{f'(N) - n} < 1/(8M)$;
\item $f''(N) < 1/(8 M^2)$.
\end{inparaenum}
\end{lemma}
\begin{proof}
Informally speaking, we first approximate $f(N+m)$ by $f(N) + f'(N)m$, then approximate $f'(N)$ by $n$, and finally apply rounding to conclude that $\nint{f}(N+m) = \nint{f}(N) + mn$. 

Formally, our first goal is to show that 
\begin{align}\label{eq:594:1}
	\abs{f(N+m) - f(N) - f'(N)m} < \frac{1}{16}.
\end{align}
Towards this end, we consider the Taylor expansion of $f$ at $N$. We have 
\begin{align*}
	f(N+m) = f(N) + f'(N)m + \frac{1}{2} f''(N+t)m^2
\end{align*}
for some $t \in [0,m]$. Thus, the expression on the right side of \eqref{eq:594:1} is
\begin{align*}
	\abs{f(N+m) - f(N) - f'(N)m} = \abs{ \frac{1}{2} f''(N+t)m^2} \leq \frac{M^2}{2} \abs{f''(N+t)}.
\end{align*}
Since $f(x) \prec x^2$, we know that $f''(x) \prec 1$, meaning that $f''(x) \to 0$ as $x \to \infty$. As a consequence, $f''$ is eventually decreasing, and picking $x_1$ sufficiently large we may assume that $f''$ is decreasing on $[N,\infty)$. Thus, condition (1) allows us to simplify our estimate to
\begin{align*}
\abs{f(N+m) - f(N) - f'(N)m} \leq \frac{M^2}{2} \abs{f''(N)} \leq \frac{1}{16},
\end{align*}
where the last inequality follows from condition (4). This completes the proof of \eqref{eq:594:1}.

In order to approximate $f'(N)$ by $n$ we simply use condition (3), which immediately yields
\begin{align*}
	\abs{f'(N)m - nm} < \frac{1}{8}.
\end{align*}
Combining this with \eqref{eq:594:1}, we get
\begin{align}\label{eq:594:3}
	\abs{f(N+m) - f(N) - nm} < \frac{1}{16} + \frac{1}{8} < \frac{1}{4}.
\end{align}

It remains to deal with rounding. We have
\begin{align*}
	 a(N+m)\!-\!a(N)\!-\!nm &= \nint{f(N+m)}\!-\!\nint{f(N)}\!-\!nm 
	 \nint{f(N+m)\!-\!\nint{f(N)}\!-\!nm }.
\end{align*}
Bearing in mind that $\nint{x} = 0$ for all $x$ with $\abs{x} < 1/2$, in order to show that $\nint{f}(N+m) = \nint{f}(N) + nm$, it suffices to estimate that
\begin{align*}
\abs{ f(N+m) - \nint{f(N)} - nm  } &\leq \frac{1}{4} + \abs{ f(N) +  nm - \nint{f(N)} - nm} 
\\ &= \frac{1}{4} + \fpa{f(N)} < \frac{1}{2}
\end{align*}
Note that the first inequality follows from \eqref{eq:594:3}, and the last one from condition (2).

\end{proof}

In the following lemma, we show how the conditions of Lemma \ref{lem:lambda-conditions-new} hold sufficiently often.

\begin{lemma}
\label{lem:lambda-sufficient}

	Let $f \colon [x_0,\infty) \to \RR$ be a Hardy field function with $x \prec f(x) \prec x^2$. Then for any real $x_1 \geq x_0$ and $\e_0,\e_1,\e_2 > 0$ there exists integer $n_0$ such that for each $n \geq n_0$ there exists integer $N$ satisfying the following conditions:
\begin{inparaenum}
\item $N \geq x_1$;
\item $\fpa{f(N)} < \e_0$;
\item $\abs{f'(N) - n} < \e_1$;
\item $f''(N) < \e_2$.
\end{inparaenum}
\end{lemma}
\begin{proof}
	Decreasing $\e_0,\e_1,\e_2$ if necessary (in this order), we may freely assume that $\e_0,\e_1,\e_2 < 1/10$ and we have the inequalities:
\begin{align*}
	\e_1 &< \frac{10}{3} \e_0,&
	\frac{\e_1}{\e_2} &> \frac{25}{\e_1} + 5 > 50.
\end{align*}
(The motivation behind these inequalities will become apparent in the course of the argument). Picking a sufficiently large value of $n_0$ we may freely assume that there exists $X_0 \geq x_1$ such that $f'(X_0) = n$, $f''(X_0) < \e_2$, and $f,f',f''$ are strictly monotone on $[X_0,\infty)$. Let $X_1,X_2,X_3$ be specified by
	\begin{align*}
	f'(X_1) &= n + \frac{1}{5} \e_1,&
	f'(X_2) &= n + \frac{2}{5} \e_1,&
	f'(X_3) &= n + \frac{3}{5} \e_1.&	
	\end{align*}
	Let also $N_1 = \ceil{X_1} < X_1 + 1$ and $N_3 = \floor{X_3} > X_3 - 1$. 
	Note that $f'(X_1) - f'(X_0) = \e_1/5$ and $f''(x) \leq f''(X_0) \leq \e_2$ for $x \in [X_0,X_1]$, so the mean value theorem implies that
	\begin{align*}
	X_1 - X_0 \geq  \frac{\e_1}{5\e_2} \geq 10.
	\end{align*}		
	By the same token, we have $X_2-X_1, X_3-X_2 \geq 10$, which in particular implies that $N_3-N_1 \geq 10$ (which is not strictly speaking necessary, but ensures that all the points under consideration appear in the order we expect). 	
	For all integers $N$ with $N_1 \leq N \leq N_3$ we have:
\begin{align*}
N &\geq x_1,&
n+\frac{1}{5}\e_1 &\leq f'(N) \leq n+\frac{3}{5}\e_1,&
0 < f''(N) \leq \e_2,&
\end{align*}	
which implies that $N$ satisfies three out of the four desired conditions.

It remains to show that there exists $N$ with $N_1 \leq N \leq N_3$ such that $\fpa{f(N)} < \e_0$. We will in fact show that for each arc $A \subset \RR/\ZZ$ with length $2\e_0$ there exists $N$ with $N_1 \leq N \leq N_3$ such that $f(N) \bmod 1 \in A$. Consider the function $\tilde f(N) = f(N) - nN$. Since $\tilde f(N) \equiv f(N) \bmod 1$, in the previous condition we may equally well require that $\tilde f(N) \bmod 1 \in A$. By the mean value theorem, for $N_1 \leq N < N_3$, the value of the expression 
\begin{align*}
	\tilde f(N+1) - \tilde f(N) = f(N+1) - f(N) - n 
\end{align*}
belongs to the interval $\left[\frac{1}{5} \e_1, \frac{3}{5} \e_1 \right]$.
Since $3\e_1/5 < 2\e_0$, it follows that, for each interval $I$ of length $2\e_0$ contained in the interval $[f(N_1),f(N_3)]$, there exists some $N$ with $N_1 \leq N \leq N_3$ with $\tilde f(N)$ intersecting $I$. By the mean value theorem, we have 
\begin{align*}
	N_3 - N_1 \geq X_3 - X_1 - 2 &\geq 
	\frac{f'(X_3) - f'(X_1)}{\max_{X_1 \leq x \leq X_3} f''(x) } - 2 \\
	&= \frac{f'(X_3) - f'(X_1)}{ f''(X_1) } - 2 \geq  \frac{2\e_1}{5\e_2} - 2 \geq \frac{10}{\e_1}.
\end{align*}
 
As a consequence, using the mean value theorem yet again, we get
\begin{align*}
	f(N_3) - f(N_1) &\geq (N_3-N_1) \min_{N_1 \leq x \leq N_3} f'(x) 
	\\ & = (N_3-N_1) f'(N_1) 
	\geq \frac{10}{\e_1} \cdot \frac{\e_1}{5} = 2.
\end{align*}
	Thus, we can find an interval $\tilde A \subset [f(N_1),f(N_3)]$ such that $\tilde A \bmod 1 = A$, and we can find $N$ with $N_1 \leq N \leq N_3$ with $\tilde f(N) \in \tilde A$. This completes the argument.
\end{proof}

Combining Lemmas \ref{lem:lambda-conditions-new} and \ref{lem:lambda-sufficient}, we immediately obtain the following consequence.

\begin{corollary}\label{cor:lambda-existence}
For each $M$ there exists $n_0$ such that for all $n \geq n_0$ there exists $N$ such that $\lambda(N,M,n)$ holds.
\end{corollary}

\subsection{$\checkmark$Restricted multiplication}
In the previous subsection, we established sufficient conditions for $\lambda(N,M,n)$ to hold. We now show how we can use $\lambda$ to define a restricted form of multiplication. Towards this end, we consider the condition $\mu_0(n,m,p)$ given by 
\begin{equation*}
\exists N \ \exists (0 \leq m < M) \ \lambda(N,M,n) \wedge (a(N+m) = a(N) + p).
\end{equation*}
Recalling that $\lambda(N,M,n)$ implies that $\nint{f}(N+m) = \nint{f}(N) + nm$, we immediately see that if $\mu_0(n,m,p)$ holds then $p = nm$. Thus, $\mu_0$ defines multiplication on the (definable) set $D_0 \subset \ZZ \times \ZZ$ consisting of the pairs $(n,m)$ such that $\mu_0(n,m,p)$.Corollary \ref{cor:lambda-existence} now translates into the following statement.

\begin{corollary}\label{cor:mu-existence}
For each $m \in \NN$ there exists $t_m$ such that $[t_m,\infty) \times \{m\} \subset D_0$.
\end{corollary}

Our next goal is to extend the definition of multiplication from $D_0$ to all of $\ZZ \times \NN$. 

\subsection{$\checkmark$Difference sets and total multiplication}

Finally, we show that we can extend the restricted multiplication defined by $\mu_0$ using difference sets. Informally, we rely on the fact that once we have defined $n m$ and $n' m$ we can define $(n-n')m$ by distributivity. To this end, define the following formula:
\begin{align*}
\mu_1 (n,m,p) &\defeq \exists n',n'',p',p''\ (\mu_0(n',m,p') \wedge \mu_0(n'',m,p'')\\\nonumber &\wedge n = n''-n' \wedge p = p''-p').
\end{align*}
This lets us prove the following lemma.

\begin{lemma}
\label{lem:general-multiplication}
Let $m \in \NN$ and $n,p \in \ZZ$. Then the formula $\mu_1(n,m,p)$ holds if and only if $nm = p$.
\end{lemma}
\begin{proof}
For the rightwards direction of proof, we know already that if $\mu_0(n,m,p)$ holds then $nm = p$. Given that $n'm = p'$ and $n''m=p''$, with $n = n''-n'$ and $p = p''-p'$, we can show by a routine computation that $nm = p$. So the rightwards direction of proof is shown.

For the leftwards direction of proof, given that $nm = p$, we fix $m$. By Corollary \ref{cor:mu-existence}, we know there exists some natural $t_m$ such that $[t_m, \infty) \times \{m\} \subseteq D_0$. Thus we know that multiplication on pairs $(n_1,m)$ is defined by $\mu_0$ when $n_1 > t_m$. So take $n' = t_m$ and $n'' = t_m + n$, and likewise take $p' = n'm$ and $p'' = n''m$. It follows that $\mu_0(n',m,p')$ and $\mu_0(n'',m,p'')$ both hold, and further that $n = n''-n'$ and $p = p''-p'$ hold. Therefore $\mu_1(n,m,p)$ holds, as required.
\end{proof}

Once multiplication is defined on $\NN \times \ZZ$ is it immediate to extend it to $\ZZ \times \ZZ$. Thus, by Lemma \ref{lem:general-multiplication}, we can defined multiplication over the whole domain of $\ZZ \times \ZZ$ in the theory $\theory{\ZZ; <, +, \nint{f}}$. It follows that the theory $\theory{\ZZ; <, +, \nint{f}}$ is undecidable, completing the proof of Theorem \ref{thm:main}.

\section{$\checkmark$Proof of the sub-linear case}

We now move to prove Theorem \ref{thm:other}. Recall that here we assume $f$ grows at least as fast as $x^\e$ for some $\e > 0$ but more slowly than $x$. Using this, we define a function which approximates $\inv{f}(x-(1/2))+(1/2)$. Given that the latter function is a Hardy field function which grows faster than $x$, we apply Theorem \ref{thm:main} and show that the theory $\theory{\ZZ; <, +, f}$ is undecidable.

Note that throughout we use the fact that, if $f$ is a Hardy field function, then so is $\inv{f}$; we take this fact from \cite[Theorem 1.7]{AD-2004}.

\begin{proof}[Proof of Theorem \ref{thm:other}] 
Let $f \colon [x_0, \infty) \to \RR$ be a Hardy field function where
\begin{eqnarray*}
\label{eqn:sub-linear-growth}
 x^{\e} \prec f(x) &\text{ and } & f(x) \prec x.
\end{eqnarray*}
both hold for some $\e > 0$.

Define the function $b(m) = \min \set{ n \in \ZZ }{ \nint{f}(n) \geq m}$, which behaves roughly like the inverse of $\nint{f}$. Formally, we define $b(m)$ in our theory as the unique value of $n$ satisfying the formula
\begin{equation*}
\nint{f}(n) \geq m \wedge \forall k\ (\nint{f}(k) \geq m \rightarrow k \geq n).
\end{equation*}
Increasing the value of $x_0$ if necessary, we may freely assume that $f$ is strictly increasing.
Additionally, we define the function $g \colon [y_0, \infty) \to \RR$ by $g(y) = \inv{f}(y - (1/2)) + (1/2)$, where $y_0$ is sufficiently large for the above definition to be well-posed. The main lemma we need to prove Theorem \ref{thm:other} is the following.

\begin{lemma}
\label{lem:agreement-of-inverses}
There exists an integer $m_0$ such that $b(m) = \nint{g(m)}$ for all integers $m \geq m_0$.
\end{lemma}
\begin{proof}
Recall that we assume that $f \colon [x_0, \infty) \to \RR$ is a strictly increasing function. Under this assumption, the inverse $\inv{f} : [f(x_0), \infty) \to \RR$ can be characterised as  $\inv{f}(y) = \min \set{x \in \RR}{f(x) \geq y}$. For an integer $m \geq m_0$ and real $x$ we have $f(x) \geq m-1/2$ if and only if $\nint{f}(x) \geq m$, directly from the definition of the nearest integer function. Thus, we know that
\begin{align*}
 \set{x \in \RR}{\!\nint{f}(x) \geq m} &= \set{x \in \RR}{\!f(x) \geq\!m-\frac{1}{2}} 
 \\ &= \set{\!x \in \RR}{\!x \geq f^{-1}\bra{m-\frac{1}{2}}} = \set{\!x \in \RR}{\!x \geq g(m) - \frac{1}{2}}.
\end{align*} 
Restricting our attention to integers, we conclude that
\begin{align*}
	\set{n \in \ZZ}{ n \geq b(m)} &=
	\set{n \in \ZZ}{ \nint{f(n)} \ge m} 
	\\ & = \set{n \in \ZZ}{n + \frac{1}{2} \geq g(m)} 
	= \set{n \in \ZZ}{n \geq \nint{g(m)}}.	
\end{align*}
 Thus, we have $b(m) = \nint{g(m)}$, as needed.
\end{proof}

Now, $g$ is a Hardy field function, as $\inv{f}$ is.  Thus, Theorem \ref{thm:main} can be applied to show the theory $\theory{\ZZ; <, +, g}$ undecidable. We do this using the following lemma.

\begin{lemma}
\label{lem:fast-growing-inverse}
We have $g(x) \prec x$ and $g(x) \succ x^{1/\e}$.
\end{lemma}
\begin{proof}
Increasing $x_0$ if necessary, we may freely assume that $f$ is strictly increasing on $[x_0,\infty)$. 
We know that for each $\delta > 0$ there exists $n_0 \geq x_0$ such that $n \geq n_0$ we have $f(n) < \delta n$. Taking $m_0 = f(n_0)$, it follows that for all $m \geq m_0$ we have $\inv{f}(m) > (1/\delta) m$. Thus, $\inv{f}(x) \prec x$ and consequently $g(x) \prec x$, as needed. The proof that $g(x) \succ x^{1/\e}$ is entirely analogous.
\end{proof}

By Lemma \ref{lem:fast-growing-inverse} the theory $\theory{\ZZ; <, +, g}$ is undecidable. Since by Lemma \ref{lem:agreement-of-inverses} we can define $\nint{g}$ in $\theory{\ZZ; <, +, \nint{f}}$, it follows that the theory $\theory{\ZZ; <, +, \nint{f}}$ is undecidable as well. This concludes the proof of Theorem \ref{thm:other}.
\end{proof} 

\bibliographystyle{alphaabbr}
\bibliography{bibliography}

\begin{thebibliography}{AvdD04}

\bibitem[AvdD04]{AD-2004}
M.~Aschenbrenner and L.~van~den Dries.
\newblock Asymptotic differential algebra.
\newblock {\em Contemporary Mathematics}, 373:49--85, 2004.

\bibitem[B{\`e}s01]{Bes-2001}
A.~B{\`e}s.
\newblock A survey of arithmetical definability.
\newblock Number suppl., pages 1--54. 2001.
\newblock A tribute to Maurice Boffa.

\bibitem[Bos81]{Boshernitzan-1981}
M.~D. Boshernitzan.
\newblock An extension of {H}ardy’s class $l$ of ``orders of infinity''.
\newblock {\em Journal d’Analyse Mathématique}, 39(1):235–255, December
  1981.

\bibitem[Bos94]{Boshernitzan-1994}
M.~D. Boshernitzan.
\newblock Uniform distribution and {H}ardy fields.
\newblock {\em J. Anal. Math.}, 62:225--240, 1994.

\bibitem[Chu36]{Church-1936}
A.~Church.
\newblock An unsolvable problem of elementary number theory.
\newblock {\em American Journal of Mathematics}, 58(2):345--363, 1936.

\bibitem[Fra09]{Frantzikinakis-2009}
N.~Frantzikinakis.
\newblock Equidistribution of sparse sequences on nilmanifolds.
\newblock {\em J. Anal. Math.}, 109:353--395, 2009.

\bibitem[GKP90]{GKP-1990}
R.~L. Graham, D.~E. Knuth, and O.~Patashnik.
\newblock {\em Concrete mathematics}.
\newblock Addison-Wesley, Reading, Mass ;, 1990.

\bibitem[KM24]{KoniecznyMullner-2024}
J.~Konieczny and C.~M{\"{u}}llner.
\newblock Bracket words along {H}ardy field sequences.
\newblock {\em Ergodic Theory Dynam. Systems}, 44(9):2621--2648, 2024.

\bibitem[Kon24]{Konieczny-2024}
J.~Konieczny.
\newblock Decidability of extensions of presburger arithmetic by generalised
  polynomials, 2024.

\bibitem[Kor01]{Korec-2001}
I.~Korec.
\newblock A list of arithmetical structures complete with respect to the
  first-order definability.
\newblock volume 257, pages 115--151. 2001.
\newblock Weak arithmetics.

\bibitem[Pre29]{Presburger-1929}
M.~Presburger.
\newblock {\"U}ber die {V}ollst\"andigkeit eines gewissen systems der
  {A}rithmetik ganzer {Z}ahlen, in welchem die {A}ddition als einzige
  {O}peration hervortritt.
\newblock {\em Comptes Rendus du I congres de Mathematiciens des Pays Slaves},
  pages 92--101, 1929.

\bibitem[Ric22]{Richter-2022}
F.~K. Richter.
\newblock Uniform distribution in nilmanifolds along functions from a hardy
  field.
\newblock {\em Journal d’Analyse Mathématique}, 149(2):421–483, December
  2022.

\bibitem[Sem80]{Semenov-1979}
A.~L. Semenov.
\newblock On certain extensions of the arithmetic of addition of natural
  numbers.
\newblock {\em Mathematics of the USSR-Izvestiya}, 15(2):401, apr 1980.

\bibitem[Sem84]{Semenov-1983}
A.~L. Semënov.
\newblock Logical theories of one-place functions on the set of natural
  numbers.
\newblock {\em Mathematics of the USSR-Izvestiya}, 22(3):587, jun 1984.

\bibitem[Tar53]{Tarski-1953}
A.~Tarski.
\newblock {\em Undecidability of the Elementary Theory of Groups}, volume~13 of
  {\em Studies in Logic and the Foundations of Mathematics}, pages 75--86.
\newblock Elsevier, 1953.

\end{thebibliography}
\end{document}